\RequirePackage{ifpdf}
\ifpdf % We are running pdfTeX in pdf mode
\documentclass[pdftex]{sigma}
\else
\documentclass{sigma}
\fi

\numberwithin{equation}{section}

\begin{document}

\allowdisplaybreaks

\renewcommand{\PaperNumber}{007}

\FirstPageHeading

\renewcommand{\thefootnote}{$\star$}

\ShortArticleName{Three Order Parameters in Quantum XZ Spin-Oscillator Models}

\ArticleName{Three Order Parameters in Quantum XZ\\ Spin-Oscillator
Models with Gibbsian Ground States\footnote{This
paper is a contribution to the Proceedings of the Seventh
International Conference ``Symmetry in Nonlinear Mathematical
Physics'' (June 24--30, 2007, Kyiv, Ukraine). The full collection
is available at
\href{http://www.emis.de/journals/SIGMA/symmetry2007.html}{http://www.emis.de/journals/SIGMA/symmetry2007.html}}}

\Author{Teunis C. DORLAS~$^\dag$ and Wolodymyr I. SKRYPNIK~$^\ddag$}

\AuthorNameForHeading{T.C.~Dorlas and W.I.~Skrypnik}

\Address{$^\dag$~Dublin Institute for Advanced Studies,  School of
Theoretical Physics,\\
$\phantom{^\dag}$~10 Burlington Road, Dublin 4, Ireland}

\EmailD{\href{mailto:dorlas@stp.dias.ie}{dorlas@stp.dias.ie}}

\Address{$^\ddag$~Institute of Mathematics of NAS of Ukraine, Kyiv,
Ukraine} \EmailD{\href{mailto:skrypnik@imath.kiev.ua}{skrypnik@imath.kiev.ua}}

\ArticleDates{Received October 29, 2007, in f\/inal form January
08, 2008; Published online January 17, 2008}

\Abstract{Quantum models on the hyper-cubic $d$-dimensional lattice of
spin-$\frac{1}{2}$ particles interacting with linear oscillators are
shown to have three ferromagnetic ground state order parameters. Two
order parameters coincide with the magnetization in the f\/irst and
third directions and the third one is a magnetization in a
continuous oscillator variable. The proofs use a generalized Peierls
argument and two Grif\/f\/iths inequalities. The class of
spin-oscillator Hamiltonians considered manifest maximal ordering in
their ground states. The models have relevance for hydrogen-bond
ferroelectrics. The simplest of these is proven to have a unique
Gibbsian ground state.}

\Keywords{order parameters; spin-boson model; Gibbsian ground state}

\Classification{82B10; 82B20; 82B26}

\def\half{\frac{1}{2}}
\def\eps{\epsilon}
\newcommand\be{\begin{gather}}
\newcommand\ee{\end{gather}}
\newcommand\non{\nonumber}
\newcommand\shalf{{\textstyle \frac{1}{2}}}

\section{Introduction}

In this paper we consider quantum lattice models of oscillators
interacting with spins whose~variables are indexed by the sites of a
hyper-cube $\Lambda$ with the f\/inite number of sites $|\Lambda|$ in
the hyper-cubic lattice $\mathbb Z^d$. Interaction is considered to be
short-range and translation inva\-riant. The~corresponding Hamiltonian
$H_{\Lambda}$ is expressed in terms of the oscillators variables
$q_{\Lambda}=(q_x, x\in \Lambda)\in {\mathbb R}^{|\Lambda|}$ and spin
$\frac{1}{2}$ Pauli matrices $S^l_{\Lambda}=(S^l_x, x\in \Lambda,
l=1,3)$, def\/ined in the tensor product of the
$2^{|\Lambda|}$-dimensional Euclidean space and the space of square
integrable functions ${\mathbb L}^2_{\Lambda}=(\otimes \mathbb
E^{2})^{|\Lambda|}\otimes L^2({\mathbb R}^{|\Lambda|})$,
\begin{gather}
H_{\Lambda}=\sum\limits_{x\in \Lambda}\big[-\partial_x^2+\mu^2 (q_x+\eta
\phi_x(S^3_{\Lambda}))^2-\mu\big]+ \sum\limits_{A\subseteq \Lambda}J_A
S^1_{[A]}+V_{\Lambda},\qquad \mu\geq 0, \quad \eta\in \mathbb
R,\label{eq1}
\end{gather}
where $\partial_x$ is the partial derivative in $q_x$, $J_A$ and
$V_{\Lambda}$ are real-valued measurable functions, the f\/irst of
which depends on $q_{\Lambda}$ and the second on $S^3_{\Lambda}$,
$q_{\Lambda}$ and the translation invariant $\phi_x$ is given by
(see Remark~\ref{remark2} in the end of the paper)
\begin{gather}
\phi_x(s_{\Lambda})=\sum\limits_{A\subseteq\Lambda} J_0(x;A)
s_{[A]}, \qquad J_0(x;x)=1.% \tag{$1.1'$}
\label{eq1'}
\end{gather}
For products of operators, functions and variables we use the
following notation: $B_{[A]}=\prod\limits_{x\in A}B_x$. The scalar
product in $(\otimes {\mathbb E}^{2})^{|\Lambda|}$,  ${\mathbb
L}^2_{\Lambda}$ will be denoted by $(\cdot,\cdot)_0$,
$(\cdot,\cdot)$, respectively. The
Schwartz space of test functions on ${\mathbb R^n}$ will be denoted by ${\mathbb S}({\mathbb R}^n)$.

We require that the Hamiltonian is well def\/ined and bounded from
below on $(\otimes \mathbb E^{2})^{|\Lambda|}\otimes
C_0^{\infty}(\mathbb R^{|\Lambda|})$, i.e.\ the tensor product of the
$2^{|\Lambda|}$-dimensional Euclidean space and the space of
inf\/initely dif\/ferentiable functions with compact supports. The
ground state average for an observable (operator) $F$ is
determined by
\[
\langle F\rangle_{\Lambda}=Z_{\Lambda}^{-1}(\Psi_{\Lambda},
F\Psi_{\Lambda}), \qquad
Z_{\Lambda}=(\Psi_{\Lambda},\Psi_{\Lambda})=||\Psi_{\Lambda}||^2,
\]
where $Z_{\Lambda}$ is a partition function. For partial cases of
$F$ we have
\[
\langle
\hat{q}_{[A]}\phi_{[A']}(S^3_{\Lambda})\rangle_{\Lambda}=Z_{\Lambda}^{-1}\int
q_{[A]} (\Psi_{\Lambda}(q_{\Lambda}),\phi_{[A']}(S^3_{\Lambda})
\Psi_{\Lambda}(q_{\Lambda}))_0dq_{\Lambda},
\]
where the integration is performed over ${\mathbb R}^{|\Lambda|}$ and
$\hat{q}_{[A]}$ is the operator of multiplication by $q_x$.

We will employ the orthonormal basis
$\psi^0_{\Lambda}(s_{\Lambda})$ of the Euclidean space $(\otimes
\mathbb E^2)^{|\Lambda|}$, diagonali\-zing~$S^3_{\Lambda}$, which is
chosen in the following way: $
\psi^0_{\Lambda}(s_{\Lambda})=\otimes_{x\in
\Lambda}\psi^0(s_x)$, $s_x=\pm 1$, $\psi^0(1)=(1,0)$, $\psi^0(-1)=(0,1)$, $
S^1\psi^0(s)=\psi^0(-s)$, $S^3\psi^0(s)=s\psi^0(s)$.
 For $F\in \mathbb L^2_{\Lambda}$ we have the following decomposition
\begin{gather*}
F(q_{\Lambda})=\sum\limits_{s_{\Lambda}}F(q_{\Lambda};s_{\Lambda})
\psi^0_{\Lambda}(s_{\Lambda}),
\end{gather*}
where the summation is performed over the $|\Lambda|$-fold Cartesian
product $(-1,1)^{|\Lambda|}\!$ of the set~$(-1,1)$.

The Hamiltonians in \eqref{eq1} are employed in hydrogen-bond
ferroelectric crystal models, consi\-de\-red in \cite{[K], [VS],[KP]}, and describe interaction between heavy ions
(oscillators with constant frequency) and protons (spins). The
second term with $J_A=0$, $|A|\geq 2$ corresponds to the energy of
protons, tunneling  along hydrogen bonds from one well to another,
and $J_x$ is associated with the tunneling frequency. The last term
in the expression for our Hamiltonian $V_A=\sum\limits_{A\subseteq
\Lambda}J_1(A)S^3_{[A]}$ describes many-body interaction between
protons ($J_1(A)$ is the intensity of the $|A|$-body interaction).

A rigorous analysis of a mean-f\/ield version of the Hamiltonian in~\eqref{eq1} with $\phi_x(S^3_{\Lambda})$ given in~\eqref{eq1'} that is
linear in $S^3$, $J_x\not=0$, $V_{\Lambda}=-(\mu
\eta)^2\sum\limits_{x\in \Lambda}\phi_x^2(S^3_{\Lambda})$,  $J_A=0$
for $|A|\geq 2$ and  $J_0(x;y)$ uniformly in lattice sites
proportional to $|\Lambda|^{-1}$ was carried  out in \cite{[KP]} in
the framework of the Bogolyubov approximating Hamiltonian method
\cite{[BBKTZ]} and occurrence of spin and oscillator orderings (the
corresponding order parameters are non-zero) for two-body
interaction between protons at non-zero temperatures was proved. To
establish such orderings for ground states without the
mean-f\/ield limit in a general case is an important task for a
theory.

The oscillator and spin orderings are established if one proves the
existence of ferromagnetic oscillator and spin long-range orders
(lro's) in ground states for the corresponding Hamiltonians. This
means that the ground state averages $\langle
\hat{q}_x\hat{q}_y\rangle_{\Lambda}$, $\langle
S^j_xS^j_y\rangle_{\Lambda}$, $j=1,3$ are bounded uniformly in
$\Lambda$ from below by positive numbers . Occurrence of the
ferromagnetic lro's implies the existence of the spin order
parameters (magnetizations in the f\/irst and third directions) $
M^l_{\Lambda}=|\Lambda|^{-1}\sum\limits_{x\in \Lambda}S^l_x$, $l=1,3$,
and the oscillator order parameter
$Q_{\Lambda}=|\Lambda|^{-1}\sum\limits_{x\in \Lambda}q_x $ in the
thermodynamic limit ($\Lambda\rightarrow \mathbb Z^d$) since the
ground state averages of their squares are uniformly bounded in
$\Lambda$ from below by a positive number.

In this paper we f\/ind functions $V_{\Lambda}$, depending on
oscillator variables, for which  ground states or eigenstates
$\Psi_{\Lambda}$ of the Hamiltonians in \eqref{eq1} are Gibbsian
\begin{equation}
\Psi_{\Lambda}(q_{\Lambda})=\sum\limits_{s_{\Lambda}}
e^{-\frac{1}{2}U(s_{\Lambda};q_{\Lambda})}
\psi^0_{\Lambda}(s_{\Lambda})\psi_{0\Lambda}(q_{\Lambda}),
\label{eq3}
\end{equation}
with the linear in $q_{\Lambda}$ spin-oscillator quasi-potential
energy $U$
\begin{equation}
U(s_{\Lambda};q_{\Lambda})=2\mu \eta \sum\limits_{x\in
\Lambda}q_x\phi _x(s_{\Lambda})+U^0(s_{\Lambda}),\label{eq4}
\end{equation}
where  $\psi_{0\Lambda}(q_{\Lambda})= \prod\limits_{x\in
\Lambda}\psi_0(q_x)$, $\psi_{0}(q)=( \mu \pi^{-1})^{\frac{1}{4}}$
$\exp\{-\frac{\mu}{2}q^2\}$, is the ground state of the free
oscillator Hamiltonian, that is the f\/irst term in the right-hand
side of \eqref{eq1} with $\eta=0$. We prove the maximal ordering in the
corresponding systems (provided some simple conditions on~$U^0$ and
$\phi_x$ are satisf\/ied): magnetizations $ M^l_{\Lambda}$, $l=1,3$,
$Q_{\Lambda}$ are non-zero. No other ground states are known with
such the property.

Gibbsian ground states were introduced by Kirkwood and Thomas in
\cite{[KT]} in  XZ spin-$\frac{1}{2}$ models with Hamiltonians
(linear in $S^1) $ that include the spin part of \eqref{eq1}, i.e.\ the
second and third terms, with only $J_x\not=0$ and the periodic
boundary condition (this boundary condition is not essential).
Matsui in \cite{[M1],[M2]} enlarged a class of
spin-$\frac{1}{2}$ $XZ$-type models in which Gibbsian states exist.
The method was further developed by Datta and Kennedy in
\cite{[DK]}. An application of classical spins systems for
constructing of quantum states was given in \cite{[GM]}.

In \cite{[DS]} we showed how to f\/ind $V_{\Lambda}$ for a given
Gibbsian ground state and established existence of lro's in $S^1$
and $S^3$ for a wide class of the spin-$\frac{1}{2}$ XZ models  (see
also \cite{[S1],[S2]}). This reference contains the most
simple proofs of the existence of lro's in ground states of quantum
many-body systems. A reader may f\/ind a review of the results
concerning several quantum orders in it (see also \cite{[MN],[ESI],[ESII]}).

The ground state in \eqref{eq3} can be represented in the following
equivalent form
\begin{equation}
\Psi_{\Lambda}(q_{\Lambda})=\sum\limits_{s_{\Lambda}}
e^{-\frac{1}{2}U_*(s_{\Lambda})}
\psi^0_{\Lambda}(s_{\Lambda})\psi_{0\Lambda}(q_{\Lambda}+
\eta\phi_{\Lambda}(s_{\Lambda})),\qquad U_*=U^0-\mu
\eta^2\sum\limits_{x\in \Lambda}\phi_x^2.\label{eq5}
\end{equation}

From \eqref{eq5}, orthonormality of the basis (see the beginning of
the third section) and the equalities $\int \psi_0^2(q)dq=1, \int
\psi_0^2(q)qdq=0$ it follows that
\begin{equation}
\langle\hat{q}_x\hat{q}_y\rangle_{\Lambda}=\eta^2\langle\phi_x(S^3_{\Lambda})
\phi_y(S^3_{\Lambda})\rangle_{\Lambda}=
\eta^2\langle\phi_x(\sigma_{\Lambda})
\phi_y(\sigma_{\Lambda})\rangle_{*\Lambda},\qquad  \langle
S^3_{x}S^3_y\rangle_{\Lambda}=
\langle\sigma_{x}\sigma_{y}\rangle_{*\Lambda},\label{eq6}
\end{equation}
where $\sigma_x(s_{\Lambda})=s_x$ and
\[
\langle\phi_x(\sigma_{\Lambda})
\phi_y(\sigma_{\Lambda})\rangle_{*\Lambda}=Z_{*\Lambda}^{-1}
\sum\limits_{s_{\Lambda}}e^{-U_*(s_{\Lambda})} \phi_x(s_{\Lambda})
\phi_y(s_{\Lambda}), \qquad
Z_{*\Lambda}=\sum\limits_{s_{\Lambda}}e^{-U_*(s_{\Lambda})} .
\]
Equalities in \eqref{eq6} reduce
 a calculation of averages in our quantum systems to a
calculation of averages indexed by a star in Ising models.  For a
short-range ferromagnetic potential energy $U_*$ the f\/irst Grif\/f\/iths
inequality holds: $\langle\sigma_{{[A]}}\rangle_{*\Lambda}\geq 0$.
As a result the following statement (principle) is true.

\begin{proposition}\label{proposition1} Let $J_0\geq 0$ in
\eqref{eq1'} and ferromagnetic lro occur in the Ising model with the
ferromagnetic potential energy $U_*$ given by \eqref{eq5}, that is
$\langle\sigma_x\sigma_y\rangle_{*\Lambda}>0$ uniformly in
$\Lambda$, then ferromagnetic lro occurs in oscillator variables and
$S^3$ in the quantum spin-oscillator system with the  Hamiltonian
\eqref{eq1} and ground state $\Psi_{\Lambda}$ in~\eqref{eq3}.
\end{proposition}
Usefulness of Gibbsian ground states is explained by comparative
simplicity of a proof of existence of lro. Our results show that
Gibbsian ground states are expected to appear in many quantum
spin-oscillator systems with non-trivial interactions.

Our paper is organized as follows. In Section~\ref{sec2} we
formulate our main results in a lemma and theorem. In next sections
we prove them.

\section{Main result}\label{sec2}

We establish that Gibbsian ground states exist for $J_A$ depending
on $q_{\Lambda}$ if
\begin{equation}
V_{\Lambda}=-\sum\limits_{A\subseteq \Lambda}J_A e^{-\frac{1}{2}
W_A(S^3_{\Lambda})},\qquad
W_A(S^3_{\Lambda})=U(S^{3A}_{\Lambda};q_{\Lambda})-
U(S^3_{\Lambda};q_{\Lambda}),\label{eq7}
\end{equation}
where $S_{\Lambda}^{3A}=(-S^3_{A},S^3_{\Lambda\backslash A})$.
$V_{\Lambda}$ determines an unbounded operator. Negative $J_A$
generate positive functions $V_{\Lambda}$ and this enables us to
prove the following lemma.

\begin{lemma}\label{lemma1}
Let $V_{\Lambda}$ be given by \eqref{eq7} and  have a domain
$D(V_{\Lambda})$, $J_A$ be bounded negative functions and $U$ in
\eqref{eq7} coincide with $U$ in \eqref{eq4}  Then $H_{\Lambda}$ is
positive definite and  essentially self-adjoint on the set $(\otimes\,
{\mathbb C}^{2})^{|\Lambda|}\otimes {\mathbb S}({\mathbb R}^{|\Lambda|})\cap
D(V_{\Lambda})$ that contains $(\otimes\, {\mathbb
C}^{2})^{|\Lambda|} \otimes C_0^{\infty}({\mathbb R}^{|\Lambda|})$ and
$\Psi_{\Lambda}$, given by~\eqref{eq3}, is an eigenfunction of the
Hamiltonian \eqref{eq1} with the zero eigenvalue and is its (unique)
ground state if $J_A\leq 0$ (if the uniform bound $J_{x}\leq J_- <
0$ holds).
\end{lemma}

\begin{remark}
\label{remark1} If the functions $J_A$ are only
negative then $\Psi_{\Lambda}$ is the ground state of the
self-adjoint extension of the Hamiltonian preserving positive
def\/initeness.
\end{remark}

 The most simple translation
invariant short-range $U^0$ is ferromagnetic
\begin{equation}
U^0(s_{\Lambda})=-\alpha \sum\limits_{A\subseteq\Lambda} J^0_{A}
s_{[A]},\qquad J^0_{A}\geq 0,\label{eq8}
\end{equation}
where $J^0_A=0$ for odd $|A|$, $\alpha\geq 0$.

\begin{theorem}\label{theorem2}
Let all the conditions of Lemma {\rm \ref{lemma1}} be satisfied, $J_0\geq 0$,
$J_0(x;A)=0$ for even $|A|$, $ J_0(0;1), J^0_{0,1}\geq \bar{J}$ and
$U^0$ be given by \eqref{eq8}. Then

{\rm I.}~For a sufficiently large $\beta=(\eta^2 \mu+\alpha) \bar{J}>1$
there exist ground state lro's in $S^3$ and oscillator variables for
$d\geq 2$;

{\rm II.} Let the positive constants $C$, $B_j$, $j=1,2,$ independent of
$\Lambda$, exist such that $|\phi_x(s_{\Lambda})|\leq C$
\[ W^{(j)}_{A}(s_{\Lambda})=\sum\limits_{x'\in\Lambda}
|\phi^j_{x'}(s^{A}_{\Lambda})-\phi^j_{x'}(s_{\Lambda})|\leq B_j,
\qquad j=1,2,
\]
where $|A|=2$. Then ground state ferromagnetic lro in $S^1$ occurs
in arbitrary dimension $d$.
\end{theorem}

If $U^0=0$ then $V_{\Lambda}$ from \eqref{eq7} is given by
\[
V_{\Lambda}=-\sum\limits_{A\subseteq \Lambda}J_Av_{[A]}, \hskip 20
ptv_x=\cosh u_x+S^3_x\sinh u_x, \qquad u_x=2\eta\mu
\phi_x(q_{\Lambda}).%\label{eq10}
\]
This equality follows from the equalities
\[
W_{A}(S^3_{\Lambda})=-2\sum\limits_{x\in A}u_xS^3_x, \qquad
e^{a S^3}=\cosh a +S^3 \sinh a,\qquad (S^3)^2=\it I.
\]
In the simplest case the conditions of item II of Theorem~\ref{theorem2} can be
checked without dif\/f\/iculty.

\begin{proposition}\label{proposition4}
Let  $\phi_x$ be linear in $S^3 $ in \eqref{eq1'} and
$||J_0||_1=\sum\limits_{x}|J_0(x)|<\infty$, where the summation is
performed over ${\mathbb Z}^d$. Then the conditions of item {\rm II} of
Theorem {\rm \ref{theorem2}} are satisfied.
\end{proposition}
Note that if one uses the Pauli matrices with $\frac{1}{2}$ instead
of the unity as matrix elements then~$V_{\Lambda}$ should be changed
by adding to $W_A$ the number $-|A|\ln 2$.

If one chooses the anti-ferromagnetic $U^0$ in \eqref{eq7},
specif\/ically,
\[
U^0(s_{\Lambda}) = \mu\eta^2\sum\limits_{x\in \Lambda}
\phi^2_x(s_{\Lambda}) + \alpha\sum\limits_{\langle x,y\rangle \in \Lambda} s_x
s_y, \qquad \alpha>0,
\]
where $x$, $y$ are nearest neighbors, then it can be easily proved that the spin lro in the third
direction will be anti-ferromagnetic, generating a staggered
magnetization (spins at the even and odd sublattices take dif\/ferent
va\-lues).

The interesting and important property of the Hamiltonians  with
$V_{\Lambda}$ given by \eqref{eq7} is that they are simply related
to generators of stationary Markovian processes (see, also,
\cite{[T]}). We believe that it is possible to apply the same
mathematical technique for proving existence of order and phase
transitions in equilibrium quantum systems and non-equilibrium
stochastic systems (see \cite{[S3]}).

\section[Proof of Lemma 1]{Proof of Lemma \ref{lemma1}}

For our purpose it is convenient to pass to a new representation. It
is determined by the Hilbert space of sequences of functions
$F(q_{\Lambda};s_{\Lambda})$, $s_x=\pm 1,$ which are found in the
expansion of the vector $F\in \mathbb L^2_{\Lambda}$ mentioned at the
beginning of the introduction, with the scalar product
\begin{gather*}
(F_1,F_2)=\sum\limits_{s_{\Lambda}}\int
F_1(q_{\Lambda};s_{\Lambda})F_2(q_{\Lambda};s_{\Lambda})dq_{\Lambda},
\\
(F_1(q_{\Lambda}),F_2(q_{\Lambda}))_0=\sum\limits_{s_{\Lambda}}
F_1(q_{\Lambda};s_{\Lambda})F_2(q_{\Lambda};s_{\Lambda}),
\end{gather*}
where the integration is performed over ${\mathbb R}^{|\Lambda|}$. Here
we took into account the orthonormality of the basis, i.e.\ the
equality
\[
(\Psi^0_{\Lambda}(s_{\Lambda}),\Psi^0_{\Lambda}(s'_{\Lambda}))_0=\delta
(s_{\Lambda};s'_{\Lambda})=\prod\limits_{x\in
\Lambda}\delta_{s_x,s'_x},
\]
where $\delta_{s,s'}$ is the Kronecker delta. Let
\begin{gather*}
h_x=-\partial_x^2+\mu^2 (q_x+\eta \phi_x(S^3_{\Lambda}))^2-\mu,
\\
h^0_x=-\partial_x^2+\mu^2 q^2_x-\mu=(-\partial _x+\mu
q_x)(\partial _x+\mu q_x)
\end{gather*}
and \[
h_{\Lambda}=\sum\limits_{x\in \Lambda}h_x,\qquad
h^0_{\Lambda}=\sum\limits_{x\in \Lambda}h^0_x.
\]
\[
(T_x(\phi)F)(q_{\Lambda})=\sum\limits_{s_{\Lambda}}F(q_x+\eta
\phi_x(s_{\Lambda}), q_{\Lambda\backslash
x};s_{\Lambda})\psi^0_{\Lambda}(s_{\Lambda}),
\]
then
\[
h_x=T_x(\phi)h_x^0T^{-1}_x(\phi), \qquad
h_{\Lambda}=T_{\Lambda}h_{\Lambda}^0T^{-1}_{\Lambda}.
\]
Here we
took into account that dif\/ferentiation commutes with $T_{\Lambda}$.

Our Hamiltonian is rewritten as follows
\begin{gather*}
H_{\Lambda}=h_{\Lambda}+\sum\limits_{A\subseteq
\Lambda}J_AP_{[A]},\qquad
P_A=S^1_{[A]}-e^{-\frac{1}{2}W_A(S^3_{\Lambda})}.%\label{eq11}
\end{gather*}
The remarkable fact is that the symmetric operator $P_{A}$ and the
harmonic operator $h_x$ have both common eigenvector
$\Psi_{\Lambda}$ with the zero eigenvalue (see Remark~\ref{remark3} in the end
of the paper). Note that the space of ground states of the operator
$h_x$ (eigenfunctions with the zero eigenvalue) is
$2^{|\Lambda|}$-fold degenerate since $S_x^3$ is diagonal and the
Laplacian is translation invariant. From \eqref{eq4} and the def\/inition of
$T_x(\phi)$ if follows that $T^{-1}_x(\phi)\Psi_{\Lambda}$ is equal
to $\psi_0(q_x)$ multiplied by a function independent of $q_x$
\[
T^{-1}_x(\phi)\Psi_{\Lambda}=\sum\limits_{s_{\Lambda}}
e^{-\frac{1}{2}U_*(s_{\Lambda})}
\psi^0_{\Lambda}(s_{\Lambda})\psi_{0\Lambda}(q_{\Lambda\backslash
x}+ \eta\phi_{\Lambda\backslash x}(s_{\Lambda\backslash
x}))\psi_0(q_x).
\]
Hence $h^0_xT^{-1}_x(\phi)\Psi_{\Lambda}=0$ and $h_x\Psi_{\Lambda}=
T_x(\phi)h^0_xT^{-1}_x(\phi)\Psi_{\Lambda}=0$. The proof that
$\Psi_{\Lambda}$ is an eigenvector with the zero eigenvalue of $P_A$
is inspired by our previous paper \cite{[DS]}. For simplicity we
will omit $q_{\Lambda}$ in the expression for $U$ in \eqref{eq3}.
Taking into consideration the equalities
\[
S^1_{[A]}\psi^0_{\Lambda}(s_{\Lambda})=\psi^0_{\Lambda}(s^A_{\Lambda})=
\psi^0_{\Lambda}(s_{\Lambda\backslash A}, -s_A),\qquad
S^3_x\psi^0_{\Lambda}(s_{\Lambda})=s_x\psi^0_{\Lambda}(s_{\Lambda}),
\]
we obtain
\begin{gather*}
(\psi_{0\Lambda})^{-1}P_A\Psi_{\Lambda}=
\sum\limits_{s_{\Lambda}}(\psi^0_{\Lambda} (s_{\Lambda\backslash
A},-s_A)-e^{-\frac{1}{2}W_A(s_{\Lambda})}
\psi^0_{\Lambda}(s_{\Lambda})) e^{-\frac{1}{2}U(s_{\Lambda})}
\\
\phantom{(\psi_{0\Lambda})^{-1}P_A\Psi_{\Lambda}}{}
=\sum\limits_{s_{\Lambda}}(\psi^0_{\Lambda} (s_{\Lambda\backslash
A},-s_A)e^{-\frac{1}{2}U(s_{\Lambda})}
-\psi^0_{\Lambda}(s_{\Lambda})e^{-\frac{1}{2}U(s^A_{\Lambda})})
\\
\phantom{(\psi_{0\Lambda})^{-1}P_A\Psi_{\Lambda}}{}
 =\sum\limits_{s_{\Lambda}}(e^{-\frac{1}{2}U(s^A_{\Lambda})}
-e^{-\frac{1}{2}U(s^A_{\Lambda})})\psi^0_{\Lambda}(s_{\Lambda})=0.
\end{gather*}
Here we changed signs of the spin variables $s_A$ in the f\/irst term
in the sum in $s_{\Lambda}$.

Positive def\/initeness of the Hamiltonian follows from the
following proposition.

\begin{proposition}\label{proposition5}
The operator $-P_A$  is positive definite on $D(V_{\Lambda})$.
\end{proposition}

\begin{proof} $V_{\Lambda}$ is an operator of multiplication by
inf\/inite dif\/ferentiable functions on each one-dimensional spin
subspace and its domain contains $(\otimes \mathbb
C^{2})^{|\Lambda|}\otimes C_0^{\infty}(\mathbb R^{|\Lambda|})$. This
domain coincides with the direct sum of $2^{|\Lambda|}$ copies of of
$L^2({\mathbb R^{|\Lambda|}}, e^{|q|_0}dq_{\Lambda})$, where
$|q|_0=\sum\limits_{x\in \Lambda}|q_x||\phi_x|$.

The scalar product in the Hilbert space $\mathbb L_{\Lambda}^2$ is
given by $(F_1, F_2)= \int (F_1(q_{\Lambda}), F_2(q_{\Lambda}))\,dq_{\Lambda} $, where the integration is performed over $\mathbb
R^{|\Lambda|}$. We have to show that
$-(P_{A}F(q_{\Lambda}),F(q_{\Lambda}))_0\geq 0$. Let us def\/ine the
operator
\[
P^+_{A}=e^{\frac{1}{2}U(S^3_{\Lambda})}P_{A}
e^{-\frac{1}{2}U(S^3_{\Lambda})}.
\]
It is not dif\/f\/icult to check on the basis $\psi^0_{\Lambda}$ that
\begin{gather*}
P^+_{A}=e^{-\frac{1}{2}W_A(S^3_{\Lambda})}(S^1_{[A]}-I),%\label{eq12}
\end{gather*}
where $I$ is the unit operator. Here we used the following equality
\begin{gather*}
e^{-\frac{1}{2}U(S^{3A}_{\Lambda})}S^1_{[A]}=
S^1_{[A]}e^{-\frac{1}{2}U(S^3_{\Lambda})}.%\label{eq13}
\end{gather*}

For the operator $P^+_{A}$ we have
\[
P^+_{A}F=\sum\limits_{s_{\Lambda}}(P^+_{A}F)(q_{\Lambda};s_{\Lambda})
\psi^0_{\Lambda}(s_{\Lambda})
\]
and
\[
(P^+_{A}F)(q_{\Lambda};s_{\Lambda})=-e^{-\frac{1}{2}
W_A(s_{\Lambda})}(F(q_{\Lambda};s_{\Lambda})-F(q_{\Lambda};s^A_{\Lambda})).
\]
It is convenient to introduce the new scalar product
\begin{gather*}
(F_1, F_2)_{U}=(e^{-U(S^3_{\Lambda})}F_1,
F_2)=\sum\limits_{s_{\Lambda}}\int
F_1(q_{\Lambda},s_{\Lambda})F_2(q_{\Lambda},s_{\Lambda})e^{-
U(s_{\Lambda})}dq_{\Lambda}
\\
\phantom{(F_1, F_2)_{U}}{}
=\int (F_1(q_{\Lambda}), F_2(q_{\Lambda}))^0_{U}dq_{\Lambda}=\int
(e^{-U(S^3_{\Lambda})}F_1(q_{\Lambda}),
F_2(q_{\Lambda}))_0dq_{\Lambda}.
\end{gather*}
The operator $P^+_{A}$ is
symmetric with respect to the new scalar product since
\begin{gather}
(P^+_{A}F_1(q_{\Lambda}),F_2(q_{\Lambda}))^0_{U}=
(P_{A}e^{-\frac{1}{2}U(S^3_{\Lambda})}F_1(q_{\Lambda}),
e^{-\frac{1}{2}U(S^3_{\Lambda})}F_2(q_{\Lambda}))_0.\label{eq14}
\end{gather}
It is not dif\/f\/icult to check that
\begin{gather}
-(P^+_{A}F(q_{\Lambda}),F(q_{\Lambda}))^0_{U}=
\sum\limits_{s_{\Lambda}} e^{-\frac{1}{2}
[U(s_{\Lambda})+U(s^A_{\Lambda})]}
(F(q_{\Lambda};s_{\Lambda})-F(q_{\Lambda};s^A_{\Lambda}))
F(q_{\Lambda};s_{\Lambda}) \nonumber\\
\phantom{-(P^+_{A}F(q_{\Lambda}),F(q_{\Lambda}))^0_{U}}{}
=\frac{1}{2}\sum\limits_{s_{\Lambda}} e^{-\frac{1}{2}
[U(s_{\Lambda})+U(s^A_{\Lambda})]}(F(q_{\Lambda};s_{\Lambda})-
F(q_{\Lambda};s^A_{\Lambda}))^2\geq 0.\label{eq15}
\end{gather}
Here we took into account that the function under the sign of
exponent is invariant under changing signs of the spin variables
$s_A$. From \eqref{eq14}, \eqref{eq15} it follows that $-P_{A}\geq
0$.
\end{proof}

From Proposition \ref{proposition5} it follows that
$H^1_{\Lambda}=\sum\limits_{A\subseteq \Lambda}J_AP_A$ is positive
def\/inite since for the negative~$J_A$ we have
\[
(H^1_{\Lambda}F, F)=\sum\limits_{A\subseteq \Lambda}\int
J_A(q_\Lambda)(P_AF(q_{\Lambda}),F(q_{\Lambda}))_0dq_{\Lambda}\geq
0.
\]

The fact that the Hamiltonian is essentially self-adjoint is derived
from the following proposition (see example X.9.3 in \cite{[RS]}).

\begin{proposition}\label{proposition6}\it The operator $-\Delta+\mu^2 x^2+V+(y,x)$,
where $(y,x)$ is the Euclidean scalar product of $x=x_j$,
$j=1,\dots ,n$ with the constant vector $y$,
$\Delta=\sum\limits_{j=1}^{n}\frac{\partial}{\partial x_j}$,
$\mu\not=0$ if $y\not=0$, $V$ is the operator of multiplication
by a non-negative  function $V(x)\in L^2({\mathbb R}^n, e^{-x^2}dx)$,
$x^2=(x,x)$, is essentially self-adjoint on ${\mathbb S}({\mathbb
R}^n)\cap D(V)$, i.e.\ the intersection of the Schwartz space and
the domain of $V$.
\end{proposition}

\begin{proof} The following  two inequalities are valid
\begin{gather}
\mu^2 ||x^2\psi||\leq ||(-\Delta +V+\mu^2
x^2)\psi||+2\mu^2n||\psi||,\qquad (y,x)\leq \epsilon
x^2+\frac{n}{4\epsilon}|y|^2,%\tag{$3,5'$}
\label{eq15'}
\end{gather}
where $\mu$ is a real number, $|y|=\max\limits_{j}|y|_j$, $||\cdot||$
is the scalar product in the Hilbert space of square integrable
functions. We tacitly assume that $x^2$ means the operator of
multiplication of a squared variable. The f\/irst inequality follows
from the inequalities ($\partial_j =\frac{\partial}{\partial
x_j}$)
\begin{gather*}
(-\Delta +V+\mu^2 x^2)^2 \geq (\mu^2 x^2)^2 -\mu^2(\Delta x^2+x^2
\Delta),
\\
-(\partial_j^2x_k^2+x_k^2\partial_j^2)= -(2x_k\partial_j^2
x_k+2\delta_{j,k}) \geq -2\delta_{j,k},\\ -(\Delta
x^2+x^2 \Delta)\geq-
\sum\limits_{j=1}^{n}(\partial_j^2x_j^2+x_j^2\partial_j^2).
\end{gather*}
Here we took into account that $V$ is positive and $[\partial
_j,x_k]=
\partial_j x_k-x_k\partial_j=\delta_{j,k}$. As a result
\[
||(y,x)\psi||\leq a||(-\Delta +V+\mu^2 x^2)\psi||+b||\psi||, \qquad a=\mu^{-2}\epsilon,\qquad
b=n\left(2+\frac{|y|^2}{4\epsilon}\right).
\]
For $\mu\not=0$ number $a$ can be arbitrary small and from the
Kato--Rellich theorem~\cite{[RS]} it follows that the essential
domain of $-\Delta+V+\mu^2 x^2+(y,x)$ coincides with the essential
domain of $-\Delta+V+\mu^2 x^2$. But the last one coincides with
${\mathbb S}(\mathbb R^n)\cap D(V)$ \cite[Theorem~X.59]{[RS]}.
Inequality~\eqref{eq15'} is suf\/f\/icient, also, for the proof of the
proposition in the case $\mu=y=0$ (this is explained in Example~X.3
in \cite{[RS]}).
\end{proof}

Since $S^3_{\Lambda}$ is  a diagonal operator on $(\otimes {\mathbb
C}^{2})^{|\Lambda|}$ the operator
\[
-\sum\limits_{x\in
\Lambda}\partial_x^2+\mu^2\sum\limits_{x\in
\Lambda}q_x^2+V_{\Lambda} + 2\eta\mu^2\sum\limits_{x\in
\Lambda}q_x\phi_{x}(S^3_{\Lambda})
\]
is the direct sum of $2^{|\Lambda|}$ copies of the minus
$|\Lambda|$-dimensional Laplacian plus the three functions
coinciding with the three functions in Proposition~\ref{proposition6}.  From
Proposition~\ref{proposition6} it follows that this operator is essentially
self-adjoint on the set $(\otimes {\mathbb
C}^{2})^{|\Lambda|}\otimes {\mathbb S}({\mathbb R}^{|\Lambda|})\cap
D(V_{\Lambda})$ that contains $(\otimes {\mathbb
C}^{2})^{|\Lambda|}\otimes \mathbb C_0^{\infty}({\mathbb R}^{|\Lambda|})$.
The same is true for the Hamiltonian since operator $(\eta
\mu)^2\sum\limits_{x\in \Lambda}\phi_x^2(S^3_{\Lambda})$ and the
operator depending on $S^1_{\Lambda}$ in its expression are bounded.

The operators $h_{\Lambda}$, $H^1_{\Lambda}$ are positive def\/inite on
the dense set $(\otimes {\mathbb C}^{2})^{|\Lambda|}\otimes {\mathbb
S}({\mathbb R}^{|\Lambda|})\cap D(V_{\Lambda})$ which is the
essential set for $H_{\Lambda}$. This implies that $H_{\Lambda}$ is
positive def\/inite on its domain $D(H_{\Lambda})$ and~$\Psi_{\Lambda}$ is its ground state.

\medskip
\noindent
{\bf Proof of uniqueness.}  We have to establish that
the symmetric semigroup $P_{\Lambda}^t$, generated by~$-H_{\Lambda}$, maps non-negative functions into (strictly) positive
functions (increases positivity) and this will imply that the ground
state is unique (see Theorem XIII.44 in \cite{[RS]}). We will establish
this property with the help of a perturbation expansion. The kernel
of the semigroup~$P_1^t$, generated by $h_{\Lambda}+V_{\Lambda}$, is
expressed in terms of the Feynman--Kac (FK) formula
\cite{[RS],[Gi]}
\begin{gather*}
P_1^t(q_{\Lambda},s_{\Lambda};q'_{\Lambda},s'_{\Lambda})=
\delta_{s_{\Lambda},s'_{\Lambda}}\int
P^t_{q_{\Lambda},q'_{\Lambda}}(dw_{\Lambda})
\exp\left\{-\int\limits_{0}^{t} V^+_{\Lambda}(w_{\Lambda}(\tau),
s_{\Lambda})d\tau\right\},%\label{eq16}
\end{gather*}
where
$V^+_{\Lambda}(q_{\Lambda};s_{\Lambda})=V_{\Lambda}(q_{\Lambda};
s_{\Lambda})+\sum\limits_{x\in \Lambda}[\mu^2 (q_x+\eta
\phi_x(s_{\Lambda}))^2-\mu]$,
$P^t_{q_{\Lambda},q'_{\Lambda}}(dw_{\Lambda})=\prod\limits_{x\in
\Lambda}P^t_{q_x,q'_x}(dw_x)$ is the conditional Wiener measure and
$w_{\Lambda}(t)$ is the sequence of continuous paths.  The semigroup~$P^t$  is represented as a perturbation series in powers of $V_0$
\[
V_0=-\sum\limits_{A\subseteq \Lambda}J_A S^1_{[A]}.
\]
This series is convergent in the uniform operator norm~\cite{[Ka]}
since $V_0$ is a bounded operator. Its perturbation expansion is
given by
\[
P^t=\sum\limits_{n\geq 0}P^t_n, \qquad
P^t_n=\int\limits_{0\leq \tau_1\leq \tau_2\leq \cdots \leq\tau_n\leq
t}d\tau_1 \cdots  d\tau_nP_1^{\tau_1}\prod\limits_{j=2}^{n+1}(V_0
P_1^{\tau_{j}-\tau_{j-1}}),
\]
where $\tau_{n+1}=t$ . We now use the following simple inequality
\begin{gather*}
\int_{0}^{t}|V^+_{\Lambda}(w_{\Lambda}(\tau),
s_{\Lambda})|d\tau
\\
\leq \bar{V}(w_{\Lambda}) =|J|_{\Lambda}\int_{0}^{t}\left[ \exp\left\{\alpha\bar{U}_0+2\eta
\mu\sum\limits_{x\in
\Lambda}\bar{\phi}_x|w_x(\tau)|\right\}+\sum\limits_{x\in \Lambda}(\mu^2
(|w_x(\tau)|+\eta \bar{\phi}_x)^2-\mu)\right]d\tau,
\end{gather*}
where $\bar{U}_0=\max\limits_{s_{\Lambda}}U_0(s_{\Lambda})$,
$|J|_{\Lambda}=\sup\limits_{q_{\Lambda}}\sum\limits_{A\subseteq
\Lambda}|J_{A}| $ ,
$\bar{\phi}_x=\max\limits_{s_{\Lambda}}|\phi_x(s_{\Lambda})|$. Let
\[
\bar{P}_1^t(q_{\Lambda};q_{\Lambda})=\int
P^t_{q_{\Lambda},q'_{\Lambda}}(dw_{\Lambda})
e^{-\bar{V}(w_{\Lambda})},\qquad  V_-=-J_-\sum\limits_{x\in
\Lambda}S^1_x
\]
then it follows from the positivity of the the kernel $P^t_1$ and
$V_0$  that
\begin{gather}
P^t(q_{\Lambda},s_{\Lambda};q'_{\Lambda},s'_{\Lambda})\geq
\bar{P}_1^t(q_{\Lambda};q_{\Lambda})\sum\limits_{n\geq
0}\frac{t^n}{n!} V^n_-(s_{\Lambda};s'_{\Lambda}).\label{eq17}
\end{gather}
Here we utilized the semigroup property of $\bar{P}^t$, the
inequalities $e^{a}\geq e^{-|a|}$, $V_0\geq V_-$,
\[
(V_0P_1^{\tau_{j}-\tau_{j-1}})(q_{\Lambda},s_{\Lambda};q'_{\Lambda},s'_{\Lambda})
\geq \bar{P}^{\tau_{j}-\tau_{j-1}}(q_{\Lambda};q'_{\Lambda})
V_-(s_{\Lambda};,s'_{\Lambda}).
\]
Now, it is easily proved as in \cite{[DS]} that the matrix $V_-$  is
irreducible. As a result there exists a~positive integer $n$ such
that and that $V^n_-$, has positive non-diagonal elements
\cite{[GI],[GII]}. Hence the kernel in the left-hand side of
\eqref{eq17} is positive.

\section{Order parameters}

In this section we will prove Theorem \ref{theorem2}, i.e.\ occurrence of
dif\/ferent types of lro in the considered quantum systems. In the
proof we will rely on the following basic theorem.

\begin{theorem}\label{theorem3}
Let the ferromagnetic short-range potential energy $U$ of the
classical Ising model on the hyper-cubic lattice $\mathbb Z^d$
with the partition function $Z_{\Lambda}=\sum\limits_{s_{\Lambda}}
\exp\{-\beta U(s_{\Lambda}) \}$ be given~by
\begin{gather}
U(s_{\Lambda})=-\sum\limits_{A\subseteq\Lambda} \varphi(A)
s_{[A]},\label{eq18}
\end{gather}
where $\varphi(A)\geq 0$, $s_x=\pm1$. Let, also, the uniform bound
$\varphi(x,y)\geq \bar{\varphi}>0$ for nearest neighbors $x$, $y$
hold and $\varphi(A)=0$ for $A$ with odd number of sites. Then for
a sufficiently large $\bar{\varphi}\beta>1$ and the dimension
$d\geq 2$ there is the ferromagnetic lro, that is, for the
Gibbsian two point spin average the uniform in $\Lambda$ bound
holds
\begin{gather*}
\langle \sigma_x\sigma_y\rangle_{\Lambda}>0, %\label{eq19}
\end{gather*}
where $\sigma_x(s_{\Lambda})=s_x$ and the magnetization (an order
parameter) $ M_{\Lambda}=|\Lambda|^{-1}\sum\limits_{x\in \Lambda}s_x
$ is non-zero in the thermodynamic limit $\Lambda\rightarrow \mathbb
Z^d$.
\end{theorem}
The proof of this theorem is based on an application of the
generalized Peierls principle (argument).  It will be given in the
end of this section (see, also, \cite{[S3],[S4]}. The next
theorem is the consequence of the basic theorem.

\begin{proof}[Proof of item I of Theorem \ref{theorem2}] Condition \eqref{eq1'} shows
that
\[
\sum\limits_{x\in
\Lambda}\phi^2_x(s_{\Lambda})=|\Lambda|+2\sum\limits_{ A\subseteq
\Lambda}J_2(A)s_{[A]}+ \sum\limits_{x\in
\Lambda}\left(\sum\limits_{x\not \in A\subseteq
\Lambda}J_0(x;A)s_{[A]}\right)^2.
\]
where $J_2(x\cup A)=J_0(x;A)$ and $J_2(A)=0$ for odd $|A|$. The
last term is equal to
\[
\sum\limits_{x\in \Lambda}\sum\limits_{x\not \in A\subseteq
\Lambda}J^2_0(x;A)+2\sum\limits_{x\in \Lambda}\sum\limits_{x\not
\in A_1,A_2\subseteq \Lambda}J_0(x;A_1)J_0(x;A_2)s_{[A_1\Delta
A_2]},
\]
where $s_{\varnothing}=1$ and $A_1\Delta A_2=( A_1\cup A_2)\backslash
(A_1\cap A_2)$. Due to translation invariance of interaction the
f\/irst term is bounded by $|\Lambda|\sum\limits_{A}J^2_0(0;A)$, where
the summation is performed over $\mathbb Z^d$, and this expression is
f\/inite since the interaction is short-range. Hence subtracting a
f\/inite constant proportional to $|\Lambda|$ from $U_*$ one sees that
the result admits representation \eqref{eq18} with positive~$J_*$
instead of $\varphi$ such that $J_*(A)=0$ for odd $|A|$ and for
nearest neighbors $x$, $y$ the inequality $J_*(x,y)\geq
(2\eta^2\mu+\alpha)\bar{J}$.
The basic theorem and Proposition~\ref{proposition1} imply occurrence of
ferromagnetic lro in $S^3$ and oscillator lro.
\end{proof}

\begin{proof}[Proof of item II of Theorem \ref{theorem2}] Let $
\psi_{\Lambda}(q_{\Lambda};s_{\Lambda})=
\psi^0_{\Lambda}(s_{\Lambda})\psi_{0\Lambda}(q_{\Lambda}). $ Then
\eqref{eq1'} and the def\/i\-ni\-tion of $S^1$ imply that
\[
S^1_x S^1_y \Psi_{\Lambda}(q_{\Lambda})=\sum\limits_{s_{\Lambda}}
e^{-\frac{1}{2}
U(s_{\Lambda};q_{\Lambda})}\psi_{\Lambda}(q_{\Lambda};s^{x,y}_{\Lambda})
=\sum\limits_{s_{\Lambda}} e^{-\frac{1}{2}
U(s^{x,y}_{\Lambda};q_{\Lambda})}\psi_{\Lambda}(q_{\Lambda};s_{\Lambda}).
\]
Taking into account also the orthonormality of the basis one obtains
\begin{gather*}
\langle S^1_x S^1_y \rangle_{\Lambda}
=Z^{-1}_{\Lambda}\sum\limits_{s_{\Lambda}}\int
e^{-\frac{1}{2}[U(s^{x,y}_{\Lambda};q_{\Lambda})+U(s_{\Lambda};q_{\Lambda})
]}\psi^2_{0\Lambda}(q_{\Lambda})dq_{\Lambda}
\\
\phantom{\langle S^1_x S^1_y \rangle_{\Lambda}}{} =Z^{-1}_{\Lambda} \sum\limits_{s_{\Lambda}} e^{\frac{\eta^2\mu}{4}
\sum\limits_{x'\in
\Lambda}(\phi_{x'}(s_{\Lambda})+\phi_{x'}(s^{x,y}_{\Lambda}))^2}
e^{-
\frac{\alpha}{2}[U_0(s^{x,y}_{\Lambda})+U_0(s_{\Lambda})]}
\\
\phantom{\langle S^1_x S^1_y \rangle_{\Lambda}}{}
\geq e^{-\frac{\alpha}{2} B_0}Z^{-1}_{\Lambda}
\sum\limits_{s_{\Lambda}} e^{\frac{\eta^2\mu}{4}
\sum\limits_{x'\in
\Lambda}(\phi_{x'}(s_{\Lambda})+\phi_{x'}(s^{x,y}_{\Lambda}))^2}
e^{- \alpha U_0(s_{\Lambda})} .
\end{gather*}

We also have
\begin{gather*}
\sum\limits_{x'\in
\Lambda}(\phi_{x'}(s_{\Lambda})+\phi_{x'}(s^{x,y}_{\Lambda}))^2\\
\qquad {}=
\sum\limits_{x'\in
\Lambda}[3\phi^2_{x'}(s_{\Lambda})+\phi^2_{x'}(s^{x,y}_{\Lambda})]+
2\sum\limits_{x'\in
\Lambda}\phi_{x'}(s_{\Lambda})[-\phi_{x'}(s_{\Lambda})+
\phi_{x'}(s^{x,y}_{\Lambda})]
\\
\qquad{}= \sum\limits_{x'\in
\Lambda}4\phi^2_{x'}(s_{\Lambda})+\sum\limits_{x'\in
\Lambda}[-\phi^2_{x'}(s_{\Lambda})+\phi^2_{x'}(s^{x,y}_{\Lambda})]+
2\sum\limits_{x'\in
\Lambda}\phi_{x'}(s_{\Lambda})[-\phi_{x'}(s_{\Lambda})+
\phi_{x'}(s^{x,y}_{\Lambda})]
\\
\qquad{}\geq 4\sum\limits_{x'\in \Lambda}\phi^2_{x'}(s_{\Lambda})-B_2-2CB_1.
\end{gather*}
This yields
\begin{gather*}
\langle S^1_x S^1_y \rangle_{\Lambda}\geq
e^{-\frac{\eta^2\mu}{4}(B_2+2B_1)}e^{-\frac{\alpha}{2}B_0}Z_{\Lambda}^{-1}\sum\limits_{s_{\Lambda}}
e^{\eta^2 \mu\sum\limits_{x'\in \Lambda}\phi^2_{x'}(s_{\Lambda})}
e^{- \alpha U_0(s_{\Lambda})}\\
\phantom{\langle S^1_x S^1_y \rangle_{\Lambda}}{} =
e^{-\frac{\eta^2\mu}{4}(B_2+2CB_1)}e^{-\frac{\alpha}{2}B_0}.\tag*{\qed}
\end{gather*}\renewcommand{\qed}{}
\end{proof}

\begin{proof}[Proof of Proposition \ref{proposition4}]    It is obvious that
$|\phi_x(s_{\Lambda})|\leq ||J_0||_1$ and that
\[
|\phi_{x'}(s^{x,y}_{\Lambda})-\phi_{x'}(s_{\Lambda})|=
|-2[s_{y}J_0(y-x')+s_{x}J_0(x-x')]| \leq 2[|J_0(y-x')|+|J_0(x-x')|].
\]
As a result $ W^{(1)}_{x,y}(s_{\Lambda})\leq 4||J_0||_1. $ Further
\begin{gather*}
|\phi^2_{x'}(s^{x,y}_{\Lambda})-\phi^2_{x'}(s_{\Lambda})|=
\Bigg|\Bigg[\sum\limits_{z\in
\Lambda\backslash(x,y)}J_0(z-x')s_{z}-s_{y}J_0(y-x')-
s_{x}J_0(x-x')\Bigg]^2
\\
\qquad{}-\Bigg [\sum\limits_{z\in \Lambda\backslash
(x,y)}J_0(z-x')s_{z}+s_{y}J_0(y-x')+s_{x}J_0(x-x')\Bigg]^2\Bigg|
\\
\qquad{} =\Bigg|-4\sum\limits_{z\in \Lambda\backslash
(x,y)}J_0(z-x')s_{z}(s_{y}J_0(y-x')+s_{x}J_0(x-x'))\Bigg|
\\
\qquad {} \leq 4\sum\limits_{z\in
\Lambda}|J_0(z-x')|(|J_0(y-x')|+|J_0(x-x')|)\leq4||J_0||_1
(|J_0(y-x')|+|J_0(x-x')|).
\end{gather*}
Hence $ W^{(2)}_{x,y}(s_{\Lambda})\leq 8||J_0||_1^2.$
\end{proof}

\begin{proof}[Proof of Theorem \ref{theorem3}] Let $\chi^{\pm}_x=\frac{1}{2}(1\pm \sigma_x)$
then one obtains
\[
4\langle\chi^+_x\chi^-_y\rangle_{\Lambda}=1+\langle\sigma_x\rangle_{\Lambda}-
\langle\sigma_y\rangle_{\Lambda}-
\langle\sigma_x\sigma_y\rangle_{\Lambda}.
\]
Since the systems are invariant under the transformation of changing
signs of spins the third and the second terms in the right-hand side
of last equality are equal to zero and
\[
\langle\sigma_x\sigma_y\rangle_{\Lambda}=1-4\langle\chi^+_x
\chi^-_y\rangle_{\Lambda}.
\]
Hence if
\begin{gather}
\langle\chi^+_x\chi^-_y\rangle_{\Lambda}<\frac{1}{4}\label{eq21}
\end{gather}
then the ferromagnetic lro occurs, i.e.
\begin{gather}
\langle\sigma_x\sigma_y\rangle_{\Lambda}\geq a>0,\label{eq22}
\end{gather}
where $a$ is independent of $\Lambda$. If one succeeds in proving
that there exists a positive function~$E_0(\beta)$ and positive
constants $a$, $a'$ independent of $\Lambda$ such that
\begin{gather}
\langle\chi^+_x\chi^-_y\rangle_{\Lambda}\leq
a'e^{E_0(\beta)}\label{eq23}
\end{gather}
and proves that $E_0$ is increasing at inf\/inity then
\eqref{eq21}, \eqref{eq22} will hold for a suf\/f\/iciently large inverse
temperature~$\beta$. The Peierls principle reduces the derivation of
\eqref{eq23} to the derivation of the contour bound.
\end{proof}

\noindent
{\bf Peierls principle.} {\it Let the contour bound hold
\begin{gather}
\langle\prod\limits_{\langle x,y\rangle\in
\Gamma}\chi^+_x\chi^-_y\rangle_{\Lambda} \leq
e^{-|\Gamma|E},\label{eq24}
\end{gather}
where $\langle\cdot\rangle_{\Lambda}$ denotes the Gibbs average for
the spin system confined to a compact domain $\Lambda$, $\Gamma$~is
a~set of the nearest neighbors, adjacent to the (connected) contour,
i.e.\ a boundary of the connected set of unit hypercubes centered at
lattice sites. Then \eqref{eq23} is valid with $E_0=a''E$, where
$a''$ is a positive constant independent of $\Lambda$.}

\begin{proof}[Proof of contour bound.]
Bricmont and Fontain derived the contour bound for the spin systems
with the potential energy \eqref{eq18} with the help of the second
Grif\/f\/iths \cite{[KS]} and Jensen inequalities \cite{[BF]} (see also~\mbox{\cite{[FL],[S4]}})
\begin{gather*}
\langle\sigma_{[A]}\sigma_{[B]}\rangle_{\Lambda[\Gamma]}-
\langle\sigma_{[A]}\rangle_{\Lambda[\Gamma]}\langle\sigma_{[B]}\rangle_{\Lambda[\Gamma]}\geq
0,%\label{eq25}
\qquad
\int e^{f}d\mu\geq \exp\left\{\int f d\mu\right\},\nonumber
\end{gather*}
where $d\mu$ is a probability measure on a measurable space. Their
proof starts form the inequality
\[
\chi^+_x\chi^-_y=e^{-\frac{\beta}{2}\sigma_x\sigma_y}e^{\frac{\beta}{2}
\sigma_x\sigma_y}\chi^+_x\chi^-_y \leq e^{-\frac{\beta}{2}
\sigma_x\sigma_y}\chi^+_x\chi^-_y\leq e^{-\frac{\beta}{2}
\sigma_x\sigma_y}.
\]
As a result ($\beta'=\bar{\varphi}\beta$)
\begin{gather*}
\langle\prod\limits_{\langle x,y\rangle \in
\Gamma}\chi^+_x\chi^-_y\rangle_{\Lambda}\leq \langle
e^{-\frac{\beta'}{2}\sum\limits_{\langle x,y\rangle\in
\Gamma}\sigma_x\sigma_y}\rangle_{\Lambda}=\langle e^{
\frac{\beta'}{2}\sum\limits_{\langle x,y\rangle \in\Gamma}\sigma_x\sigma_y}\rangle_
{\Lambda[\Gamma]}^{-1}
\\
\phantom{\langle\prod\limits_{\langle x,y\rangle \in
\Gamma}\chi^+_x\chi^-_y\rangle_{\Lambda}}{}
\leq e^{-\frac{\beta'}{2}\sum\limits_{\langle x,y\rangle \in\Gamma}
\langle\sigma_x\sigma_y\rangle_{\Lambda[\Gamma]}}
=e^{-E_{\Gamma}},
\end{gather*}
where $\langle\cdot ,\cdot\rangle_{\Lambda[\Gamma]}$ is the average
corresponding to the potential energy
\[
U_{\Gamma}(q_{\Lambda})=U(s_{\Lambda})+\frac{\bar{\varphi}}{2}\sum\limits_{\langle x,y\rangle \in
\Gamma}s_xs_y.
\]
In the last line we applied the Jensen inequality. From the second
Grif\/f\/iths inequality it follows that the average
$\langle\sigma_x\sigma_y\rangle_{\Lambda[\Gamma]}$ is a monotone
increasing function in $\varphi_{A}$. So, in the potential energy
determining this average we can put $\varphi_A=0$, except $A=\langle x,y\rangle$
 and leave the coef\/f\/icient $\bar{\varphi}$ in front of the bilinear
nearest-neighbor pair potential in \eqref{eq18} without increasing
the average. This leads to
\[
\langle\sigma_x\sigma_y\rangle_{\Lambda[\Gamma]}\geq
\langle\sigma\sigma'\rangle=Z^{-1}_2\left(\frac{\beta'}{2}\right)\sum\limits_{s_1,s_2=\pm
1}s_1s_2e^{\frac{\beta'}{2}s_1s_2},\qquad
Z_2(\beta)=\sum\limits_{s_1,s_2=\pm 1}e^{\frac{\beta'}{2}s_1s_2}.
\]
That is,
\[
E_{\Gamma}\geq |\Gamma| E, \qquad E=2^{-1}\beta'
\langle\sigma \sigma'\rangle
\]
or
\begin{gather*}
E=\beta (e^{2^{-1}\beta'}-e^{-2^{-1}\beta'})
(e^{2^{-1}\beta'}+e^{-2^{-1}\beta'})^{-1}\geq 2^{-1}\beta'
(1-e^{-\beta'}).%\label{eq26}
\end{gather*}
Here we used in the denominator the inequality
 $e^{-2^{-1}\beta'}\leq e^{2^{-1}\beta'}$. Obviously, $E$
tends to inf\/inity if $\beta'$ tends to inf\/inity. This implies \eqref{eq24}.
\end{proof}

\section{Discussion}
We showed that in the considered lattice
spin-boson models  with $J_A\leq 0$ ground states are Gibbsian and
the ground state averages for special observables are reduced to
averages in classical Ising models. This means that existence of
ground states order parameters is connected with existence of order
parameters in the associated Ising models and that a breakdown of
symmetries in the quantum systems is determined by a breakdown of
symmetries in Ising models. We considered the free boundary
conditions implying that for the cases of the perturbation
$V_{\Lambda}$, considered in the two theorems, the ground state
averages of $q_x$, $S^3_x$ are zero, that is  $\langle
\hat{q}_x\rangle_{\Lambda}=0$, $\langle S^3_x\rangle_{\Lambda}=0$ if
the associated Ising potential energy is an even function. In order
to make such the averages non-zero (explicit symmetry breaking) one
has to introduce special boundary conditions (quasi-averages) which
have to single out pure Gibbsian states in the associated Ising
models. It is known \cite{[Ai]} that for the two-dimensional
ferromagnetic Ising nearest-neighbor model there are two boundary
conditions which generate pure states and that every other state is
a convex linear combination of these two states. A discussion of a
construction of ground states in lattice spin and fermion quantum
systems with an explicit symmetry breaking a reader may f\/ind in~\cite{[KoT]}.

\begin{remark}
\label{remark2} Translation invariance means that
\[
J_{x_1,\dots , x_n}=J_{0,x_2-x_1,\dots ,x_n-x_1}, \hskip 20 pt
J_0(x;x_1,\dots ,x_n)=J_0(0;x_1-x,\dots ,x_n-x)
\]
where $J$, $J_0$ are symmetric functions. The short-range
character of interaction means that
\[
\max_{x}\sum\limits_{A}|J_{x,A}|< \infty,\qquad
\max_{x}\sum\limits_{A}|J_0(x;A)|< \infty.
\]
\end{remark}

\begin{remark}
\label{remark3}   If only one-point sets are
left in the sum for $V_{\Lambda}$ then the expression for
$H_{\Lambda}$ can be rewritten in the following way
\[
H_{\Lambda}=\sum\limits_{x\in\Lambda}H_{x}.
\]
The property of the ground state $\Psi_{\Lambda}$ to be a ground
state with the zero eigenvalue of a local Hamiltonian $H_x$ was
found earlier for special isotropic anti-ferromagnetic Heisenberg
chains with valence bond ground state in \cite{[AKLT]}.
\end{remark}

\pdfbookmark[1]{References}{ref}
\LastPageEnding

\end{document}